\documentclass[pdflatex,sn-mathphys]{sn-jnl}
\usepackage[utf8]{inputenc}
\usepackage{amsfonts}
\usepackage{amssymb}
\usepackage{amsmath}
\usepackage{amsthm}
\usepackage{graphicx}
\usepackage{color}
\usepackage{float}

\theoremstyle{thmstyleone}%
\newtheorem{theorem}{Theorem}
\newtheorem{lemma}[theorem]{Lemma}

\theoremstyle{thmstyletwo}%
\newtheorem{remark}{Remark}%

\theoremstyle{thmstylethree}%
\newtheorem{definition}{Definition}%

\begin{document}

\title[Sackin and Colless indices]{Two Results about the Sackin  and Colless Indices for Phylogenetic Trees and Their Shapes}
\author[1]{\fnm{Gary} \sur{Goh}}\email{e0148664@u.nus.edu}

\author[2]{\fnm{Michael} \sur{Fuchs}}\email{mfuchs@nccu.edu.tw}

\author*[1]{\fnm{Louxin} \sur{Zhang}}\email{matzlx@nus.edu.sg}

\affil*[1]{\orgdiv{Department of Mathematics}, \orgname{National University of Singapore}, \orgaddress{\street{10 Lower Kent Ridge Road}, \city{Singapore}, \postcode{119076}, 
\country{Singapore}}}
\affil[2]{\orgdiv{Department of Mathematical Sciences}, \orgname{National Chengchi University}, \orgaddress{
\city{Taipei}, \postcode{116}, 
\country{Taiwan}}}




\abstract{The Sackin and Colless  indices are two widely-used metrics for measuring the balance of trees and for testing evolutionary models in phylogenetics. This short paper contributes two results about the Sackin and Colless indices of trees.
One result is the asymptotic analysis of the expected Sackin and Colless indices of a tree shape (which are full binary rooted unlabelled trees) under the uniform model where tree shapes are sampled with equal probability. Another is   a short elementary proof of the closed formula for the expected Sackin index of phylogenetic trees (which are full binary rooted trees with leaves being labelled with taxa) under the uniform model.}



\maketitle

\section{Introduction}
The Sackin \cite{Sackin72Syst,Shao90Syst} and Colless  \cite{Colless82Syst} indices are two widely-used metrics for measuring the balance of phylogenetic trees and
testing evolutionary models \cite{Avino19EE,Xue20PNAS, Mooers97Quarterly,Scott20Syst,Blum06PLOS,Kirkpatrick93Evol}. 
Phylogenetic trees are binary
rooted trees in which each internal node has two children and only the leaves are labelled one-to-one with taxa.
For a phylogenetic tree, its Sackin index is defined as the sum over its internal nodes of the number of leaves below that node, whereas its Colless index is defined as the sum over its internal nodes of the  balance of that node, where the balance of a node is defined to be the difference in the number of leaves below the two  children of that node. Because of their wide applications,  the two tree balance metrics have been extensively studied in the past decades (see the recent comprehensive survey \cite{Fisher21Survey}).

The Sackin and Colless indices of a random phylogenetic tree have been investigated under the Yule-Harding  model (where tree shapes of $n$ leaves  are generated using a birth-death process and their leaves are labeled according to a permutation of taxa chosen uniformly at random) and the uniform model (where trees are sampled with equal probability) \cite{Kirkpatrick93Evol,Heard92Evol,Blum05Math,Blum06AAP}.
The expected Sackin and Colless indices of a phylogenetic tree are proved to be asymptotic to  $\sqrt{\pi}n^{3/2}$ under the uniform model and $n\log n$ under the Yule-Harding model \cite{Blum05Math,Blum06AAP}. Recently, Mir et al. 
\cite{Mir13Math} discovered surprisingly that the expected Sackin index of a phylogenetic tree is simply $\frac{4^{n-1}n!(n-1)!}{(2n-2)!}-n$ under the uniform model. An alternative proof of this closed formula was given by King and Rosenberg \cite{King21Math}. Both asymptotic and exact results on the variances of the Sackin and Colless indices have also been reported \cite{Kirkpatrick93Evol,Blum05Math,Blum06AAP,Coronado20BMC}.

It is not hard to see that the Sackin index of a binary tree is actually  equal to the sum of the depths of all its leaves \cite{Steel16Book}. Therefore, 
the Sackin index and the tree height have also been studied for other types of trees in the combinatorics and theoretical computer science literature \cite{Flajolet82JCSS,Broutin12RSA, Fill04TCS,Fuchs15JMB}. 

In this paper, we focus on two questions about the Sackin and Colless indices.
The first question is what the expected Sackin and Colless indices of a random binary  tree shape are under the uniform model \cite{Rogers96Syst}. Here, tree shapes (also called Otter or Polya trees) are binary rooted trees with unlabeled leaves where each internal node has two children. Although there is increasing interest in tree balance indices for tree shapes in the study of  phylodynamic problems \cite{colijn2018metric,kim2020distance}, to the best of our knowledge, the statistical properties of these two indices and other tree balance indices have not been formally studied for tree shapes \cite{Fisher21Survey}. 
  Here,   we prove that 
the  expected Sackin and Colless indices of a  tree shape with $n$ leaves are asymptotic to $\sqrt{\pi}\lambda^{-1}n^{3/2}$ under the uniform model, where  $\lambda \approx 1.1300337163$.

Given that the closed formula (mentioned above) for the 
expected Sackin index of a phylogenetic tree under the uniform model is rather simple, the second question is whether an elementary proof exists for the formula or not.
We answer this question by using a simple recurrence for the Sackin index that is derived using the fact that all the phylogenetic trees on $n$ taxa can be enumerated by inserting the $n$-th taxon into every edge of the phylogenetic trees on $n-1$ taxa \cite{Fel04Book}. Recently, this technique was used by Zhang  for computing the sum over all nodes of the number of the descendants of that node and counting the number of tree-child networks with 
one reticulation \cite{Zhang19BMC}. 

\section{Basic definitions and notation}

\subsection{Phylogenetic trees and shapes}

A tree shape is a full binary rooted tree in which all nodes are unlabeled. 
A phylogenetic tree on $n$ taxa is 
a full binary rooted tree with $n$ leaves in which its leaves are uniquely labeled with a taxon and each of  the $n-1$ non-leaf nodes has two children.

Let $T$ be a phylogenetic tree on $n$ taxa or a tree shape. We use $V_0(T)$ to denote the set of all non-leaf nodes of $T$ and $V(T)$ to denote the set of all nodes. 
A leaf $x$  is said to be below a node $u$ in $T$ if the unique path from the root to $x$ passes through $u$. We use $\ell_T(u)$ to denote  the number of leaves below $u$ in $T$. Also, we set $\ell_T(u)=1$ if $u$ is a leaf.

Let $u \in V_0(T)$.  The {\it balance} of $u$ is defined to be 
$\lvert \ell_T(v)-\ell_T(w) \rvert$,
where $v$ and $w$ are the two children of $u$.  We use $\delta _T(u)$ to denote the balance of $u$.

For each non-root $u\in V(T)$, we use $p(u)$ to denote the parent of $u$ in $T$.

\subsection{Sackin and Colless indices}
\begin{definition}
The Sackin index of a tree shape or a phylogenetic tree $T$ is defined to be $\sum_{u\in V_0(T)} \ell_T(u)$,  and denoted by $S(T)$.
\end{definition}

\begin{definition}
The Colless index of a tree shape  or a phylogenetic tree $T$ is defined to be $\sum_{u\in V_0(T)}\delta_T(u)$,  and denoted by 
$C(T)$.
\end{definition}

The expected Sackin and Colless indices of a tree shape under the uniform model are respectively defined as: 
\[
\mbox{ESI}_{sh}(n)=\frac{1}{b_n}\sum_{T\in {\cal T}(n)}\mbox{S}(T)
\]
and 
\[
\mbox{ECI}_{sh}(n)=\frac{1}{b_n}\sum_{T\in {\cal T}(n)}C(T),
\]
where ${\cal T}(n)$ denotes the set of all tree shapes with $n$ leaves and $b_n=\vert{\cal T}(n)\vert$. 
Although there does not exist a closed formula for $b_n$,  $b_n$ can  be computed using the following recurrence formulas for $n>1$ (A001190 in the  On-Line Encyclopedia of Integer Sequences\footnote{\url{https://oeis.org/}}):
\begin{equation}\label{eqn3}
   b_n=\sum_{1\leq k <n/2} b_kb_{n-k} + \begin{cases}
    0, & \mbox{if $n$ is odd;}\\
    {\displaystyle\frac{1}{2}b_{n/2}(b_{n/2}+1)}, &  \mbox{if $n$ is even.}\end{cases}
\end{equation}
Equivalently, the generating function 
$B(z)=\sum_i b_iz^i$ satisfies the following equation: 
\begin{equation}
  B(z)=z+\frac{1}{2}\left(B(z)^2+B(z^2)\right). \label{eqn_GF}
\end{equation}

The expected Sackin index of a phylogenetic tree under the uniform model is defined similarly, that is, 
\[
\mbox{ESI}_{p}(n)=\frac{1}{a_n}\sum_{P\in {\cal P}(n)}S(P),
\]
where ${\cal P}(n)$ denotes the set of all phylogenetic trees on $n$ taxa and $a_n=\vert{\cal P}(n)\vert=\frac{(2n-2)!}{2^{n-1}(n-1)!}$ (see \cite{Steel16Book}). 

%
%
%
%
%
%
\section{Asymptotic analysis of the expected Sackin and Colless indices for tree shapes}

Recall that  ${\cal T}(n)$ denotes the set of all possible tree shapes with $n$ leaves. Let $S_n=\sum_{T\in {\cal T}(n)} S(T)$, which is the sum of the Sackin index over all tree shapes with $n$ leaves. Obviously, 
$S_1=0$ and $S_2=2$. 

For $n>2$, ${\cal T}(n)$ can be obtained by combining every pair of tree shapes 
$T'\in {\cal T}(k)$ and $T''\in {\cal T}(n-k)$, where
$k$ can range from $1$ to $n/2$. 
For a specific $k\leq n/2$,  $T\in {\cal T}(k)$ and $T'\in {\cal T}(n-k)$,
$S(T)=n+S(T')+S(T'')$ for the tree shape 
$T$ obtained by combining $T'$ and $T''$, as there are $n$ leaves below the root of $T$.

Using the facts mentioned above and Eqn. (\ref{eqn3}), we obtain that:
\begin{align}
S_n &=  \sum_{1\leq k <  n/2} \left( \sum_{T\in {\cal T}(k)}\sum_{T'\in {\cal T}(n-k)} \left(n+S(T)+S(T')\right)\right) \nonumber \\
&= \sum_{1\leq k <  n/2} \left( nb_kb_{n-k}+ \sum_{T\in {\cal T}(k)}\sum_{T'\in {\cal T}(n-k)} \left(S(T)+S(T')\right)\right) \nonumber \\
&=n \sum_{1\leq k <  n/2} b_kb_{n-k} \nonumber \\
& \;\;\;+ \sum_{1\leq k <  n/2} \left( \sum_{T\in {\cal T}(k)}\sum_{T'\in {\cal T}(n-k)} S(T) + \sum_{T\in {\cal T}(k)}\sum_{T'\in {\cal T}(n-k)} S(T')\right)
\nonumber \\
&= n b_n + \sum_{1\leq k< n/2}\left( b_{n-k}S_k + b_k S_{n-k}\right)  \nonumber \\
&=n b_n
+ \sum_{1\leq k <n}  S_kb_{n-k}, 
 \label{eqn6}
\end{align}
for odd $n$ and 
\begin{align}
S_n 
&=nb_n + \sum_{1\leq k <  n/2} \left( \sum_{T\in {\cal T}(k)}\sum_{T'\in {\cal T}(n-k)} \left(S(T)+S(T')\right) \right)  \nonumber \\
&\quad + 
\sum_{T, T'\in {\cal T}(n/2): T\neq T'} \left(S(T)+S(T')\right) +\sum_{T\in {\cal T}(n/2)} 2S(T) \nonumber  \\
&= n b_{n}  + \sum_{1\leq k< n/2}\left(b_{n-k} S_k + b_{k} S_{n-k}\right)   + \left(\sum_{T\in {\cal T}(n/2)}(b_{n/2}-1)S(T)\right)  + 2S_{n/2}  \nonumber \\
&= nb_{n}
+ S_{n/2} + \sum_{1\leq  k <n} S_kb_{n-k}
 \label{eqn7}
\end{align}
for even $n$.

\subsection{The asymptotic value of $\mbox{ESI}_{sh}(n)$}

It is unknown whether or not one can derive a closed formula for $S_n$  from Eqn. (\ref{eqn6})-(\ref{eqn7}).  However, 
 an asymptotic analysis of $S_n$  follows from the classical asymptotic analysis of $b_n$ from Eqn. (\ref{eqn3}). In order to recall the latter, we need the notion of $\Delta$-analyticity. First, a $\Delta$-domain with parameters $\delta$ and $\phi$  is a domain in the complex plane of the form:
\[
\Delta=\{z\in{\mathbb C}\ :\ \vert z\vert<1+\delta,\ \vert\arg(z-1)\vert>\phi\}
\]
with $\delta>0$ and $0<\phi<\pi/2$; see Definition VI.1 in \cite{Flajolet09Book}. A function is called $\Delta$-analytic if it is analytic in such a $\Delta$-domain. 

\begin{lemma} (\cite{Broutin12RSA})
The convergence radius $\rho$ of the generating function $B(z)$ of $b_n$ in
Eqn. (\ref{eqn_GF}) satisfies $1/4\leq \rho\leq1/2$, where  $\rho+B(\rho^2)/2=1/2$. Moreover, $B(z)$ is $\Delta$-analytic and satisfies as $z\rightarrow\rho$ in {\color{blue} a} $\Delta$-domain:
\begin{equation}\label{sing-exp-Bz}
B(z)=1-\lambda\sqrt{1-z/\rho}+\mathcal{O}(1-z/\rho),\qquad\lambda=\sqrt{2\rho+2\rho^2B'(\rho^2)}.
\end{equation}
Thus,
\begin{equation}\label{exp-Bn}
b_n\sim\frac{\lambda}{2\sqrt{\pi}n^{3/2}\rho^{n}},\qquad (n\rightarrow\infty).
\end{equation}
\end{lemma}

\begin{remark} $\rho$ and $\lambda$ can be computed up to very high precision, e.g.,
\[
\rho=0.40269750367\cdots\qquad\text{and}\qquad\lambda= 1.1300337163\cdots.
\]
The computation is done as follows: first, use Eqn. (\ref{eqn3}) to compute a truncated version $\tilde{b}_n$ of $b_n$; then use it to compute a truncated version $\tilde{B}(z)$ of $B(z)$; finally, find $\tilde{\rho}$ with $\tilde{\rho}+\tilde{B}(\rho^2)/2=1/2$. Clearly, $\tilde{\rho}$ approximates $\rho$ and this approximation can be made arbitrarily precise; also an approximation of $\lambda$ can be derived from it via Eqn. (\ref{sing-exp-Bz}).
\end{remark}

\begin{remark}\label{trans-thm}
The asymptotic expansion in Eqn. (\ref{exp-Bn}) follows from the {\it singularity expansion} in Eqn. (\ref{sing-exp-Bz}) by the {\it transfer theorems} (see Theorem VI.3 and Corollary VI.1 in \cite{Flajolet09Book}) which assert that if $A(z)$ is $\Delta$-analytic with  $A(z)\sim c(1-z/\rho)^{-\alpha}$, where $c,\rho\in{\mathbb R}\setminus\{0\}$ and $\alpha\in{\mathbb C}\setminus\{0,-1,-2,\ldots\}$, then $[z^n]A(z)\sim [z^n]c(1-z/\rho)^{-\alpha}\sim c\rho^{-n}n^{\alpha-1}/\Gamma(\alpha)$, where $[z^n]f(z)$ denotes the $n$-th coefficient in the Maclaurin series of $f(z)$ and $\Gamma(z)$ is the gamma function. More generally, the process of showing that $A(z)$ is $\Delta$-analytic, deriving the expansion $A(z)\sim c(1-z/\rho)^{-\alpha}$ as $z\rightarrow\rho$ and then using the transfer theorems to obtain the asymptotics of $[z^n]A(z)$ is called {\it singularity analysis}; see Chapter VI in \cite{Flajolet09Book}.
\end{remark}

\begin{remark}\label{clos-prop}
 Singularity analysis is closed under several operations on functions; see Section VI.10 in \cite{Flajolet09Book}. For instance, if singularity analysis can be applied to $A(z)$, it can also be applied to $A'(z)$, where the singularity expansion of $A'(z)$ is obtained from the one of $A(z)$ by term-by-term differentiation. E.g., $B'(z)$ from the previous lemma is also $\Delta$-analytic with singularity expansion as $z\rightarrow\rho$
\begin{equation}\label{sing-exp-diff-B}
B'(z)\sim\frac{\lambda}{2\rho}\cdot\frac{1}{\sqrt{1-z/\rho}},
\end{equation}
from which the asymptotic expansion of $[z^n]B'(z)$ follows by the transfer theorems.  (Of course, since $[z^n]B'(z)=(n+1)[z^{n+1}]B(z)$, this expansion is just the expansion in Eqn. (\ref{exp-Bn}) multiplied by $n/\rho$.)
\end{remark}


\begin{theorem}  
\label{res-Sackin}
Under the uniform model, the  
expected Sackin index of a tree shape with $n$ leaves,
$\mbox{\rm ESI}_{sh}(n)$, is  asymptotic to  $\pi^{1/2}\lambda^{-1}n^{3/2}$, where $\lambda$ is given in Eqn. (\ref{sing-exp-Bz}).
\end{theorem}
\begin{proof}
The recurrence formulas in Eqn. (\ref{eqn6})-(\ref{eqn7}) translate into the following equation for the generating function $S(z)=\sum_{i}S_iz^i$ of $S_n$:
\begin{equation}\label{eq-S}
S(z)=zB'(z)+S(z)B(z)+S(z^2)
\end{equation}
since the generating function of $\sum_{1\leq k<n}S_kb_{n-k}$ is the product $S(z)B(z)$ and
\[
\sum_{n\geq 1}nb_nz^n=zB'(z),\qquad \sum_{n\ \text{even}}S_{n/2}z^n=S(z^2).
\]
Rewriting Eqn. (\ref{eq-S}) into
\[
S(z)=\frac{zB'(z)+S(z^2)}{1-B(z)},
\]
 we deduce that the radius of convergence of $S(z)$ is equal to $\rho$. Moreover, from Eqn. (\ref{sing-exp-Bz}) and the closure properties of singularity analysis (Remark~\ref{clos-prop} above), we obtain that $S(z)$ is $\Delta$-analytic and satisfies as $z\rightarrow\rho$ in a $\Delta$-domain:
\[
S(z)\sim\frac{\rho(\lambda/2\rho)(1-z/\rho)^{-1/2}+S(\rho^2)}{\lambda\sqrt{1-z/\rho}+{\mathcal O}(1-z/\rho)}\sim\frac{1}{2}\cdot\frac{1}{1-z/\rho},
\]
where we used Eqn. (\ref{sing-exp-diff-B}) and $\rho<1$ which implies that $S(z^2)$ is analytic at $z=\rho$.

By the transfer theorems (see Remark~\ref{trans-thm}), we obtain: 
\begin{equation}\label{asymp-Sn}
S_n\sim\frac{1}{2}[z^n](1-z/\rho)^{-1} \sim\frac{1}{2\rho^{n}},\qquad (n\rightarrow\infty)
\end{equation}
and thus 
\[
\mbox{ESI}_{sh}(n)=\frac{S_n}{b_n} 
\sim
\frac{1/(2\rho^n)}{\lambda/\left(2\sqrt{\pi}n^{3/2}\rho^n\right)} = \sqrt{\pi}\lambda^{-1}n^{3/2},\qquad (n\rightarrow\infty)
\]
using Eqn. (\ref{exp-Bn}).
This proves the claim.
\end{proof}

\subsection{The asymptotic value of $\mbox{ECI}_{sh}(n)$}

Next, we derive the asymptotic value of $\mbox{ECI}_{sh}(n)$. First, for each internal node $u$ of a tree, we use  $c_1(u)$ and $c_2(u)$ to denote the two children of $u$. We have that
$\ell(u)=\ell(c_1(u))+\ell(c_2(u))$ and thus 
$\delta(u)=\lvert \ell(c_1(u))-\ell(c_2(u))\rvert =\ell(u)-2\min(\ell(c_1(u)), \ell(c_2(u))$. From this, it follows that for each tree shape $T$, 
$D(T)=S(T)-C(T)=2\sum_{u\in V_0(T)} \min(\ell(c_1(u)), \ell(c_2(u)).$

Defining
$$D_n=\frac{1}{2} \sum_{T\in {\cal T}(n)} D(T),$$  we obtain:
\begin{equation}\label{Colless-Sackin}
 C_n=\sum_{T\in{\cal T}(n)} C(T)=S_n- 2D_n.
\end{equation}
In addition, we have the following recurrence formula: 

\begin{align}
D_n &=\sum_{1\leq k <  n/2} \left( 
\sum_{T\in {\cal T}(k)}\sum_{T'\in {\cal T}(n-k)} \left(D(T)+D(T')+k\right)\right) \nonumber \\
&= \sum_{1\leq k < n/2}k b_k b_{n-k}
+ \sum_{1\leq k< n/2}\left( b_{n-k}D_k +
 b_k D_{n-k}\right)  \nonumber \\
 &= \sum_{1\leq k \leq  n/2}kb_kb_{n-k}
+ \sum_{1\leq  k <n}  D_kb_{n-k}, 
\;\; \mbox{for odd $n$}\nonumber
\end{align}
and 
\begin{align}
D_n &=\sum_{1\leq k <  n/2} \left( 
\sum_{T\in {\cal T}(k)}\sum_{T'\in {\cal T}(n-k)} (D(T)+D(T')+k\right)  \nonumber \\
&  \quad + 
\sum_{T, T'\in {\cal T}(n/2): T\neq T'} [D(T)+D(T')+n/2] +\sum_{T\in {\cal S}(n/2)} [2D(T)+n/2] \nonumber  \\
&= \sum_{1\leq k < n/2}k b_kb_{n-k}
+ \sum_{1\leq k< n/2}\left(b_{n-k} D_k +
 b_{k} D_{n-k}\right)  \nonumber \\
 &\quad + \left(\sum_{T\in {\cal T}(n/2)}(b_{n/2}-1)D(T)\right) + {b_{n/2}\choose 2}\frac{n}{2} + 2D_{n/2}+\frac{n}{2} b_{n/2}  \nonumber \\
 &= \sum_{1\leq k \leq  n/2}k b_kb_{n-k} + \sum_{1\leq k <n} D_kb_{n-k}   -\frac{n}{2}{b_{n/2}\choose 2}
+ D_{n/2},
\;\; \mbox{for even $n$.}\nonumber
\end{align}

 We first need a technical lemma for:
\[
F_n:=\sum_{1\leq k\leq n/2}kb_kb_{n-k}+\begin{cases} 0,&\text{if $n$ is odd;}\\ {\displaystyle-\frac{n}{2}\binom{b_{n/2}}{2}},&\text{if $n$ is even}.\end{cases}
\]


\begin{lemma}
We have
$F_n={\mathcal O}\left(n^{-1}\rho^{-n}\right).
$
\end{lemma}
\begin{proof}
By using Eqn. (\ref{exp-Bn}),
\begin{align*}
F_n&= {\mathcal O}\left(\rho^{-n}\sum_{1\leq k\leq n/2}k^{-1/2}(n-k)^{-3/2}+n^{-2}\rho^{-n}\right)\\
&={\mathcal O}\left(n^{-1}\rho^{-n}\int_{0}^{1/2}x^{-1/2}(1-x)^{-3/2}{\rm d}x+n^{-2}\rho^{-n}\right)\\
&={\mathcal O}\left(n^{-1}\rho^{-n}+n^{-2}\rho^{-n}\right)={\mathcal O}\left(n^{-1}\rho^{-n}\right),
\end{align*}
where in the second step, we approximated the sum by an integral.
\end{proof}

Now, define:
\begin{equation}\label{rec-tildeDn}
\tilde{D}_n=Kn^{-1}\rho^{-n}+\sum_{1\leq k<n}\tilde{D}_kb_{n-k}+\begin{cases}0,&\text{for $n$ is odd;}\\ \tilde{D}_{n/2},&\text{for $n$ is even,}\end{cases}
\end{equation}
where $K$ is the implied ${\mathcal O}$-constant from the last lemma. The reason for considering this sequence is that it (a) majorizes $D_n$, namely, $D_n\leq\tilde{D}_n$ (which is easily proved by induction) and (b) its asymptotics can derived with similar tools as used in the proof of Theorem~\ref{asym-mean-S}.

\begin{lemma}
\label{est-Cn2}
We have,
\[
\tilde{D}_n\sim\frac{K}{\lambda\sqrt{\pi}}n^{-1/2}(\log n)\rho^{-n},\qquad (n\rightarrow\infty).
\]
Consequently, $D_n={\mathcal O}\left(n^{-1/2}(\log n)\rho^{-n}\right)$.
\end{lemma}
\begin{proof}
Let $\tilde{D}(z)=\sum_{i}\tilde{D}_iz^i$ be the generating function of $\tilde{D}_n$. Then, the recurrence in Eqn. (\ref{rec-tildeDn}) translates into
\[
\tilde{D}(z)=K\log\frac{1}{1-z/\rho}+\tilde{D}(z)B(z)+\tilde{D}(z^2)
\]
since
\[
\sum_{n\geq 1}Kn^{-1}\rho^{-n}z^n=K\log\frac{1}{1-z/\rho}
\]
and the rest of terms are explained as in the derivation of Eqn. (\ref{eq-S}). Solving for $D(z)$ gives:
\[
\tilde{D}(z)=\frac{\displaystyle K\log\frac{1}{1-z/\rho}+\tilde{D}(z^2)}{1-B(z)}.
\]
Thus, from Eqn. (\ref{sing-exp-Bz}), $\tilde{D}(z)$ satisfies as $z\rightarrow\rho$ in a $\Delta$-domain:
\[
\tilde{D}(z)\sim\frac{\displaystyle K\log\frac{1}{1-z/\rho}+\tilde{D}(\rho^2)}{\lambda\sqrt{1-z/\rho}+{\mathcal O}(1-z/\rho)}\sim\frac{K}{\lambda}\cdot\frac{\displaystyle \log\frac{1}{1-z/\rho}}{\sqrt{1-z/\rho}}
\]
from which the claimed result follows by the transfer theorems (which also work with $\log$-factors; see Theorem VI.3 in \cite{Flajolet09Book}).
\end{proof}

Now from Eqn. (\ref{asymp-Sn}), Eqn. (\ref{Colless-Sackin}) and Lemma~\ref{est-Cn2}, we have the following result.

\begin{theorem}\label{asym-mean-S}
Under the uniform model, the expected Colless index of a tree shape with $n$ leaves, $\mbox{\rm ECI}_{sh}(n)$, is asymptotic to  $\pi^{1/2}\lambda^{-1}n^{3/2}$.
\end{theorem}

\begin{figure}
\centering 
\includegraphics[scale=0.50]{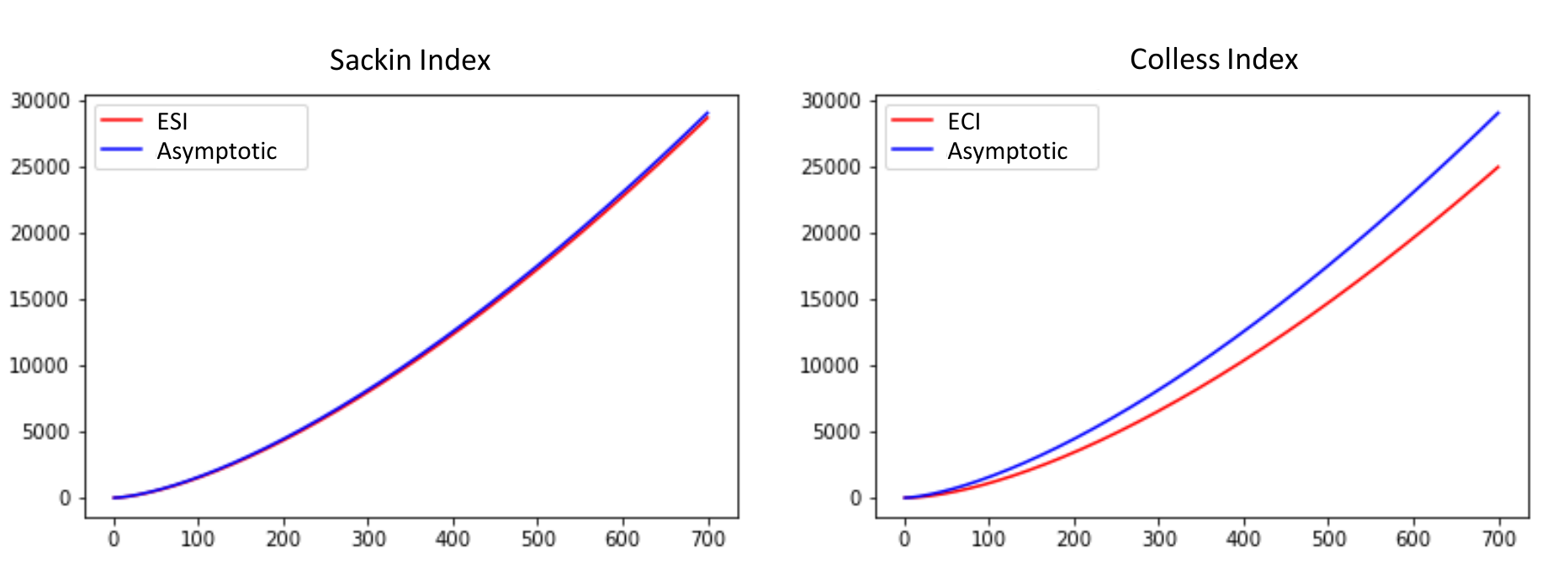}
\caption{The exact and asymptotic values of the expected Sackin (left) and Colless (right) indices.}
\label{Fig_Comparison}
\end{figure}

\subsection{Visualization on the asymptotic analyses} 

The exact and asymptotic values of 
$\mbox{ESI}_{sh}(n)$ and $\mbox{ECI}_{sh}(n)$ were computed and compared for $n$ up to 700 (Figure~\ref{Fig_Comparison}). The comparison indicates that the asymptotic value $\sqrt{\pi}\lambda^{-1}n^{3/2}$ is a very good approximation to the Sackin index even for a small number $n$. However,  the asymptotic value overestimates the Colless index with a relatively large margin. The large margin is due to the fact that $\mbox{ESI}_{sh}(n)-
\mbox{ECI}_{sh}(n)$ is of the order $n\log n$ according to our proof; however, the relative error will tend to $0$ with a speed of at least $\log n/\sqrt{n}$. 

\section{The expected Sackin index for phylogenetic trees} 

Mir et al. discovered the following simple closed formula for the expected Sackin index for a phylogenetic tree under the uniform model.  

\begin{theorem}
\label{Thm1}
(\cite{Mir13Math}) For any $n$, $\mbox{\rm ESI}_{p}(n)=\frac{4^{n-1}n!(n-1)!}{(2n-2)!}-n$.
\end{theorem}

An alternative proof was presented in \cite{King21Math} recently. Here,  we will present a short elementary proof using the following enumeration of phylogenetic trees (see \cite{Fel04Book} for example):
\begin{quote}
Assume that there is an open edge entering the root of each phylogenetic tree.   ${\cal P}(n+1)$
can be obtained from ${\cal P}(n)$ by attaching Leaf $n+1$ on each of the $2n-1$ edges of every tree of ${\cal P}(n)$ (Figure~\ref{Fig1_attach}.A). 
\end{quote}

\begin{figure}
\centering 
\includegraphics[scale=0.50]{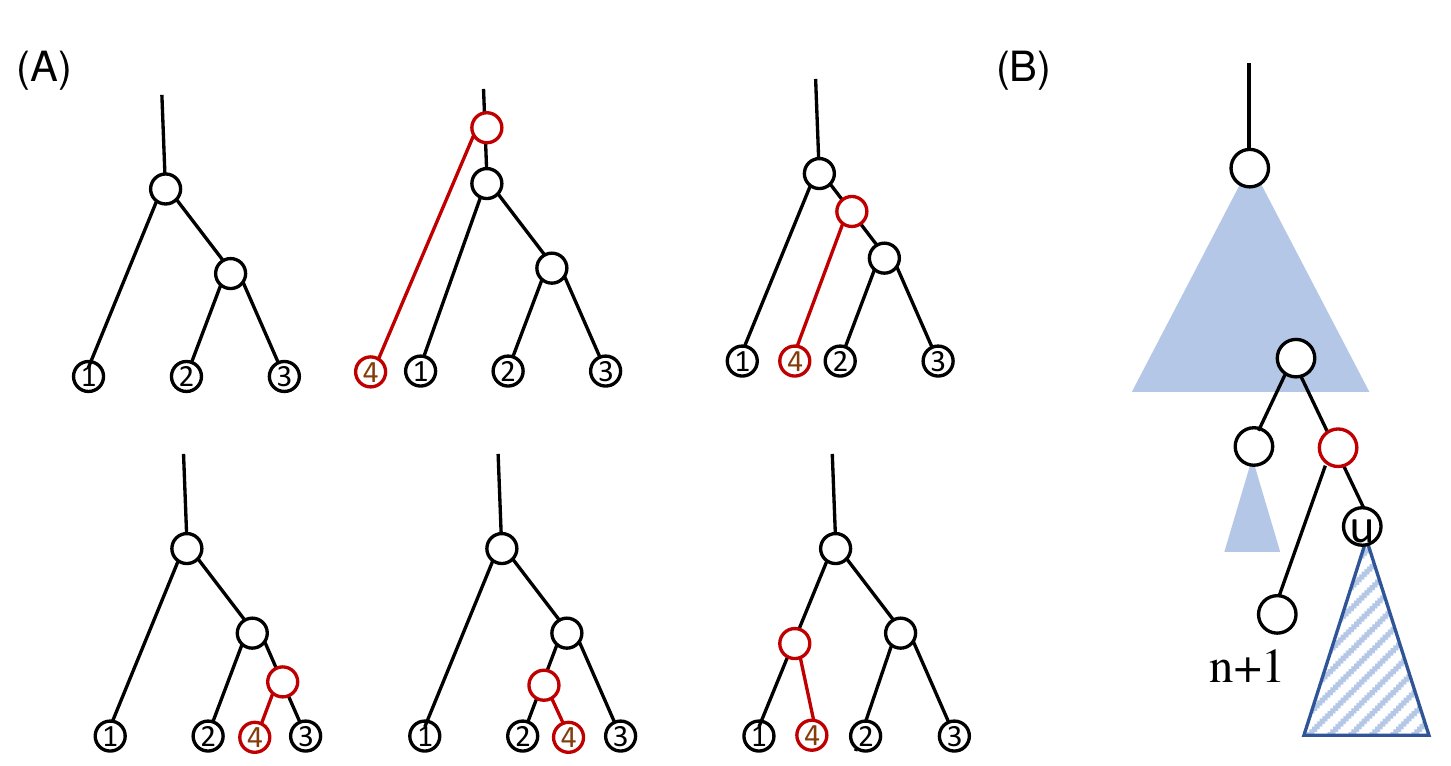}
\caption{(A) Illustration of the process of generating phylogenetic trees on $n+1$ taxa through inserting Leaf $n+1$ in each edge of a
phylogenetic tree on taxa $\{1, 2, \cdots, n\}$ for $n=3$. (B)
After Leaf $n+1$ is attached onto the edge entering the node $u$, the number of leaves below
the parent of $n+1$ in the obtained tree $Q$ is equal to 1 plus that of $u$ in the original tree $P$. }
\label{Fig1_attach}
\end{figure}

Let $S^{(p)}_{n}=\sum_{P\in {\cal P}(n)} S(P)$. Note that 
$S^{(p)}_{n}=\mbox{ESI}_p(n)\times a_n$, where $a_n=\vert{\cal P}(n)\vert$. For each $P\in {\cal P}(n)$, we use ${\cal A}(P)$ to denote the set of $2n-1$
phylogenetic trees on $n+1$ taxa that are obtained from $P$ by attaching Leaf $n+1$ on each of the $2n-1$ tree edges of $P$.
Then,
\begin{equation}
   \label{last_recurrence}
    S^{(p)}_{n+1}=\sum_{P\in {\cal P}(n)}
    \sum_{Q\in {\cal A}(P)} S(Q).
\end{equation}

Consider a  tree $Q\in {\cal A}(P)$. 
Note that Leaf $n+1$ and its parent are the only nodes of $Q$ that are not found in $P$.  Assume that $Q$ is obtained by attaching Leaf $n+1$ to the edge $e$ that enters $u$ in $P$. The number of leaves below  the parent of Leaf $n+1$  is $1+\ell_P(u)$ in $Q$ 
(Figure~\ref{Fig1_attach}.B).
Therefore, the amount contributed by the parents of Leaf $n+1$ to the sum $\sum_{Q\in {\cal L}(P)}S(Q)$ is:
\begin{align}
\sum_{u \in V(P)}(1+\ell_P(u))&= (2n-1)+\sum_{u \in V(P)} \ell_P(u) \nonumber \\ 
&= (2n-1)+n+\sum_{u \in V_0(P)}\ell_P(u)  \nonumber \\
&= 3n-1+S(P),
\end{align}
where the $n$ in the second expression is the sum of $\ell_P(u)$ (which is $1$) over all the $n$ leaves $u$ in $P$.

For $w\in V(P)$, we have either  $\ell_P(w) = \ell_Q(w)$ or 
$\ell_P(w) = \ell_Q(w)+1$. Furthermore, the latter holds if and only if $Q$ is obtained by attaching Leaf $n+1$ to an edge below $w$ in $P$. Since there are $2\ell_P(w)-2$ edges below $w$ in $P$,  thus $\ell_Q(w)= \ell_P(w)+1$ for exactly 
$2\ell_P(w)-2$ trees $Q$ of ${\cal A}(P)$. 
Therefore, 
\begin{align*}
  \sum_{Q\in {\cal A}(P)} S(Q) &= (2n-1)S(P)
   +[S(P)+(3n-1)] + \sum_{w\in V_0(P)} 
   (2\ell_P(w)-2)\\
 &= 2nS(P) + (3n-1) + 2S(P) -2\lvert V_0(P)\rvert \\
 &= 2(n+1)S(P) + (n+1).
 \end{align*}
 Adding $n+1$ to each term in the left-hand side of the above equality, which can be considered as the contribution of the $n+1$ leaves, we further have:
 \begin{align*}
  \sum_{Q\in {\cal L}(P)}\left(S(Q)+(n+1)\right) &=
 2(n+1)S(P) + (n+1) + (2n-1)(n+1)\\ 
 &= 2(n+1)\left(S(P)+n\right).
 \end{align*}
By Eqn.~(\ref{last_recurrence}), we obtain the following simple recurrence formula: 
\begin{align}
 S^{(p)}_{n+1}+(n+1)a_{n+1} &=
 \sum_{P\in {\cal P}(n)}
\sum_{Q\in {\cal L}(P)}\left(S(Q)+(n+1)\right) \nonumber  \\
&=  \sum_{P\in {\cal P}(n)} 2(n+1)\left(S(P)+n\right) \nonumber \\
 &= 2(n+1) \left(
 S^{(p)}_n+na_{n} \right). \label{zzz}
\end{align}
Since $S^{(p)}_2=2$ and $a_2=1$, Eqn.~(\ref{zzz}) implies that 
$S^{(p)}_n=2^{n-1}n!-na_n$ and 
$$
\mbox{ESI}_p(n)=\frac{S^{(p)}_n}{a_n} = \frac{4^{n-1}n!(n-1)!}{(2n-2)!}-n
$$
Theorem~\ref{Thm1} is proved.

\section{Conclusion}
In this short paper, we contributed two results to the study of the Sackin and Colless indices.
We have proved that the asymptotic value of Sackin and Colless indices are the same for tree shapes under the uniform model. Note that this is expected since tree shapes under the uniform model are known to behave similar to phylogenetic trees under the uniform model; see the discussion in the introduction of \cite{Broutin12RSA}. In particular, the average height of phylogenetic trees and binary tree shapes with $n$ leaves are both asymptotically equal to $2\lambda^{-1}\sqrt{\pi n}$ 
(see \cite{Flajolet82JCSS} and \cite{Broutin12RSA}).

We also presented a short elementary proof of the closed formula for the expected Sackin index of phylogenetic trees under the uniform model. 
The proof is based on a  tree enumeration approach that is  different from one used in  \cite{Mir13Math} and \cite{King21Math}. This technique  was also used by Goh \cite{Goh_Thesis} to derive  a short proof of the closed formula for the expected total cophenetic index of a phylogenetic tree under the uniform model that was introduced in \cite{Mir13Math} (see also \cite{Fisher21Survey}).
It is an interesting problem  whether or not the proof technique in Section 4  can be used to investigate other tree balance indices (such as those given in the survey paper \cite{Fisher21Survey}). 
\section*{CRediT authorship contribution statement}
G. Goh: Recurrence formulas;
L. Zhang:  Recurrence formulas, writing;
M. Fuchs: Asymptotic analysis, writing.

\section*{Declaration of competing interest}
The authors declare that they have no known competing financial
interests or personal relationships that could have appeared to
influence the work reported in this paper.

\section*{Acknowledgments}
The authors thanks the two anonymous reviewers for useful suggestions and comments for preparing the final version of this paper. 
LZ was supported by MOE Tier 1 grant R-146-000-318-114; MF was supported by MOST-109-2115-M-004-003-MY2.

\bibliography{bibliography}

\end{document}